\documentclass{article}

\usepackage[english]{babel}

\usepackage[letterpaper,top=2cm,bottom=2cm,left=3cm,right=3cm,marginparwidth=1.75cm]{geometry}

\usepackage{amsmath}
\usepackage{amsfonts}
\usepackage{amsthm}
\usepackage{amssymb}
\usepackage{graphicx}
\usepackage[colorlinks=true, allcolors=blue, breaklinks=true]{hyperref}
\usepackage{booktabs}
\usepackage[authoryear]{natbib}
\usepackage{authblk}
\usepackage{xcolor}
\usepackage{ulem}
\newtheorem{proposition}{Proposition}[section]

\def\keywords{\vspace{.5em}
{\noindent\textbf{Keywords}:\,\relax%
}}

\def \qmo{``}

\def \qmcsp{'' }
\def \bs{\mathbf}

\DeclareMathOperator*{\argmin}{arg\,min}

\title{Quantile mixed graphical models with an application to mass public shootings in the United States}
\author[1]{Luca Merlo}
\author[2,3]{Marco Geraci}
\author[2]{Lea Petrella}

\affil[1]{Department of Human Sciences, European University of Rome, Italy \authorcr luca.merlo@unier.it}
\affil[2]{MEMOTEF Department, Sapienza University of Rome, Italy}
\affil[3]{Department of Epidemiology and Biostatistics, University of South Carolina, USA}

\date{\today}

\begin{document}
\maketitle

\begin{abstract}
Over the last fifty years, {\color{black}the United States have} experienced hundreds of mass public shootings {\color{black}that} resulted in thousands of {\color{black}victims}. {\color{black}Characterized by their frequent occurrence} and devastating nature, mass shootings have become a major public health hazard {\color{black}that} dramatically impact safety and well-being of individuals and {\color{black}communities.} 
 {\color{black} Given the epidemic traits of this phenomenon,} there have been concerted efforts to understand the root causes that lead to public mass shootings in order to implement effective prevention strategies. {\color{black}We propose} a quantile mixed graphical model {\color{black}for investigating the intricacies of inter- and infra-domain relationships of this complex phenomenon,} {\color{black}where} conditional relations {\color{black}between discrete and continuous variables are modeled without stringent distributional assumptions} using Parzen's definition of mid-quantile. To retrieve the graph structure and recover only the most relevant connections, we consider the neighborhood selection approach in which conditional mid-quantiles of each variable in the network are modeled as a sparse function of all others. We propose a two-step procedure to estimate the graph where, in the first step, conditional mid-probabilities are obtained semi-parametrically and, in the second step, the model parameters are estimated by solving an implicit equation with a LASSO penalty.
\end{abstract}

\section{Introduction}\label{sec:intro}
A violent crime known as a mass shooting occurs when an assailant uses a firearm to kill or injure a number of individuals {\color{black}within the same firing episode}. {\color{black}There are different definitions based on several criteria and thresholds that qualify an incident with multiple victims as a mass shooting.} Consistently with the {\color{black}convention} agreed upon by criminologists and the Congressional Research Service \citep{CRS}, {\color{black}as well as by the maintainers of the Violence Project Database (VPD)}, throughout this paper we define a \emph{mass public shooting} (MPS) as
{\color{black}
\begin{quote}
a multiple homicide incident in which four or more victims are murdered with firearms--not including the offender(s)--within one event, and at least some of the murders occurred in a public location or locations in close geographical proximity (e.g., a workplace, school, restaurant, or other public settings), and the murders are not attributable to any other underlying criminal activity or commonplace circumstance (armed robbery, criminal competition, insurance fraud, argument, or romantic triangle).
\end{quote}
}

Over the last {\color{black} fifty or so years, in the United States (US)} there have been more than 180 {\color{black}MPSs}, the most deadly {\color{black}of which} occurred {\color{black} in the new millennium}. This phenomenon has caused a dramatic loss of human life and has heavily {\color{black}wreaked havoc the lives of many more people}. The ongoing backdrop of violence has deleterious effect{\color{black}s not only} at the individual level, {\color{black}in that} the sense of {\color{black}personal} well-being and safety {\color{black}is progressively eroded in people's perceptions}, but also at the institutional and societal levels {\color{black}as the proposed solutions and policies surrounding this problem have polarized the public opinion}. In addition, it is also alarming that US mass shootings are becoming more frequent {\color{black} and} deadlier. The availability of {\color{black}assault and} semi-automatic weapons with large-capacity magazines and high-volume rapid fire {\color{black}has contributed to sharply increase} the death count per shooting {\color{black}in the last two decades}. In response to {\color{black}the} recent spike {\color{black}in MPSs}, authorities and policy makers have been {\color{black}called on} to identify the {\color{black}key} factors behind mass shootings and {\color{black}to offer} solutions to prevent them from happening.

Previous articles in the literature have investigated the associations between measures of gun availability, gun ownership and firearm violence \citep{lin2018have, reeping2019state}. \cite{lankford2020have} documented how public mass shootings became more deadly over time due to societal changes that have led to more shooters motivated by fame or attention, as well as to those who have been directly influenced by previous attackers. In the wake of shootings in schools such as those at Columbine, Virginia Tech and Uvalde, a number of researchers have sought to identify the features that school shooters have in common in terms of family life, personalities, histories, and behaviors \citep{langman2009rampage, katsiyannis2018historical, peterson2021presence}. Very recently, \cite{stoffi2023} employed causal inference approaches to assess the effect of distance between schools and gun retailers and the risk of gun incidents in schools. {\color{black}Other common} civilian {\color{black} (thus vulnerable and generally unprotected) targets} {\color{black}include shopping centres and places of work, often because of a direct and personal connection of the offenders with those sites} \citep{capellan2021investigation}. There are also a number of studies {\color{black}in} psychiatry, psychology, public health, and sociology, that investigated the possible reasons behind the rise in mass shooting incidents. \cite{fox1998multiple} wrote that the motives for mass murder are organized around five primary themes (revenge, power, loyalty, terror, profit) that can occur singly or in combination, indicating a detailed level of planning and mental preparation. \cite{fox2003mass} then outlined a number of common demographics, psychological and behavioral features in the profile of a mass shooter. {\color{black}The} assailants, almost all of which are male, tend to share characteristics from past trauma and personal crises, including depression, resentment, social isolation, the tendency to externalize blame, fascination with graphically violent entertainment, and a keen interest in weaponry. As discussed in \cite{fox2014mass}, however, these indicators may constitute warning signs but tend to over-predict potential perpetrators as they are prevalent in a large portion of the population. Indeed, \cite{metzl2015mental} and \cite{metzl2021mental} have more recently suggested that mental illness alone does not provide sufficient evidence to explain and prevent mass shootings. In general, focusing on individual domains tends to ignore the important personal, physical and social contexts surrounding the attack \citep{hirschtritt2018reassessment}. On the contrary, the mental health and psychological dimensions are complexly interrelated with {\color{black}the assailant's} personal {\color{black}and familial} history, {\color{black}social and economic networks, and} local gun culture. As discussed in \cite{peterson2021violence}, there is no single profile of a shooter or predictor of mass shootings that can reduce the problem to a unidimensional perspective. If these complexities are not properly taken into account, the understanding of the real extent of these tragic events, and consequently how to effectively prevent them, will be limited or misguided.

In the present article, we considered data from the Violence Project's mass shooter database in the US. The Violence Project is a nonprofit research center dedicated to public criminology and data-driven violence prevention. The database {\color{black}was} built using public records and open-source data{\color{black}. It} includes {\color{black}a number of} variables {\color{black}pertaining} to mass murders occurred in a public location from 1966 to this day, {\color{black}making it the largest and the most comprehensive to date}. This database is {\color{black}now} at the {\color{black}base} of an extensive literature on gun violence {\color{black}evidence-based research} \citep[see, among others,][]{peterson2021communication, hoops2021consensus, jewett2022us, peterson2022psychosis}. {\color{black}When it comes to} statistical modeling, the analysis of mass shootings data poses several challenges. First, the motives for these tragic events are rooted in complex, multi-layered processes, which involve factors that pertain to different domains. Confining the analysis to univariate methods completely ignores the dependence structure with other factors and only conveys a partial view of such complex phenomenon. Second, empirical {\color{black}distributions} often exhibit {\color{black}``non-standard'' features like} heavy {\color{black}tailedness}, skewness {\color{black}and heteroskedasticity}, which cannot be accommodated {\color{black}easily} by {\color{black}standard} parametric models {\color{black}such as those based on normality assumptions}. {\color{black} Lastly, the analysis of the interdependencies of public mass shootings data involves both discrete and continuous variables, and this calls for appropriate modeling approaches.}

{\color{black}In} the present article, {\color{black}we} seek to address all these important aspects {\color{black}by developing} a new graphical modeling approach {\color{black}that} allows us to {\color{black}identify and quantify} the intricate relationships between multiple risk factors{\color{black}, whether it be} continuous {\color{black}or} discrete, {\color{black} and that is applicable under general distributional assumptions.} Graphical models have become a popular and effective tool for the statistical analysis of conditional dependence relations among variables \citep[see][]{lauritzen1996graphical, jordan1999learning, whittaker2009graphical, koller2009probabilistic}. {\color{black}In} this {\color{black}kind of analytic} framework, interrelationships {\color{black}among variables of, say, physical, biological or social phenomena} can be {\color{black}represented as networks} through a graph whose nodes correspond to variables {\color{black}while} edges connecting them {\color{black}depict} interactions.

Unfortunately, the literature {\color{black}on} graphical models {\color{black}for mixed variables} (continuous, count and categorical) is fairly limited. In parallel efforts, \cite{yang2014mixed} and \cite{chen2015selection} introduced the class of \emph{mixed graphical models} (MGMs), {\color{black}where} the univariate conditional distribution of each variable given the {\color{black}others is specified} as {\color{black}one} of the exponential {\color{black}families} of distributions. The graph structure is then estimated by fitting (regularized) \emph{generalized linear models} (GLMs) for each node to perform neighborhood selection \citep{meinshausen2006high}. Subsequently, in a related line of research, \cite{lee2015learning} and \cite{cheng2017high} proposed a generalization of the conditional Gaussian model of \cite{lauritzen1989mixed} for mixed data. However, these proposals suffer from practical limitations {\color{black}of the parametric forms of those distributions such as, for example, the inability of the normal to model effects that go beyond location and scale, or the quite stringent constraint that the single-parameter Poisson distribution imposes on the simultaneous modeling of location, scale and shape.}

In this paper, we generalize MGMs by introducing a quantile-based graphical model for mixed variables that tackles conditional dependency structures, without making assumptions on the functional form of the distributions. {\color{black} 
 The use of quantiles provides the means to study the entire conditional distribution of the network variables. Compared to moments-based networks, a quantile graphical framework may shed light on relationships that can otherwise be missed by modeling a limited number of moments, and is more apposite to describe these relationships in the tails of the distributions.} 
 In the context of undirected graphs, quantile graphical models have been proposed by \cite{ali2016multiple} and \cite{chun2016graphical} to recover conditional independencies, even for heteroskedastic, non-Gaussian data, but they are confined to the modeling of continuous variables only. While continuous responses have been the focus of the majority of progress in quantile modeling \citep{koenker2005quantile, koenker2017handbook, furno2018quantile, uribe2020quantile, merlo2023unified}, the discrete case has received {\color{black}comparatively} much less attention thus far. Quantile methods for discrete responses do, in fact, present some hurdles, such as the lack of a general theory for handling different types of discreteness, practical estimation challenges, and the troublesome asymptotic behavior of sample quantiles in the presence of ties. Here we start from the work of \cite{geraci2022mid} who developed a quantile regression method for discrete responses by extending Parzen's definition of marginal mid-quantiles (\citealp{parzen1993change}). Intuitively, mid-quantiles can be viewed as fractional order statistics and have been extensively studied by \cite{ma2011asymptotic}. In this context, using mid-quantiles comes with desirable advantages as opposed to existing approaches, based on either latent constructs \citep{kordas2006smoothed} or jittering \citep{machado2005quantiles}. Indeed, the former generally relies on maximum score estimation which is computationally expensive as it involves nonconvex loss functions, while the latter may lack generality as it requires that adjacent values in the support of the response variable are equally spaced. On the contrary, mid-quantiles are computationally efficient, well-behaved asymptotically \citep{ma2011asymptotic} and, most importantly, offer a unifying theory for quantile estimation with discrete or continuous dependent variables.
 
 Given the complexity of the phenomenon under study, {\color{black} the identification of the root causes of MPSs may require including a large number of variables, which, however, can potentially muddy the interpretation of the results.} 
 To identify only the most important conditional dependence relations and therefore induce sparsity in the network, we model the conditional mid-quantiles of each variable as a sparse function of all others and fit separate regularized regressions using the neighborhood selection method of \cite{meinshausen2006high}. For each variable, the parameters are estimated via a two-step procedure where conditional mid-probabilities are first obtained semi-parametrically using logistic regression and then regression coefficients are estimated by solving a penalized implicit equation {\color{black}based on the} least absolute shrinkage and selection operator (LASSO, \citealt{tibshirani1996regression}). The proposed method allows us to embed in a common graphical framework both continuous (possibly, e.g., heavy-tailed, skewed, {\color{black}or} multimodal) and discrete (e.g., binary, ordinal, count) variables, thus offering a much richer class of {\color{black}distribution-free} conditional distribution estimates than {\color{black}those based solely on} the conditional mean.

The rest of this paper is organized as follows. In Section \ref{sec:data} we introduce the VPD in detail. {\color{black}We formally introduce the proposed model in Section \ref{sec:met} and the estimation procedure in Section \ref{sec:est}. We assess the proposed methods in a simulation study (Section \ref{sec:sim}) and then analyze real data on MPSs in the US between 1966 and 2022 (Section \ref{sec:app}). We offer some final remarks in Section \ref{sec:con}.}

\section{The data}\label{sec:data}
{\color{black}As defined in the \texttt{theviolenceproject.org} website, the Violence Prevention Project Research Center (``The Violence Project'') is a nonpartisan nonprofit dedicated to reducing violence through research that is public-facing and informs policy and practice.} Co-founded by psychologist Dr. Jillian Peterson and sociologist Dr. James Densley, it has developed an integrated, interdisciplinary, understanding of violence and a holistic approach to addressing it. The areas of expertise cover gun violence, violent extremism, cyber violence, trauma and mental illness, street gangs and youth violence. The Violence Project conducts high-quality, high impact, research for all the public, enhancing and supporting education and training programs for schools {\color{black}and} universities, workplaces and retail stores in their strategic response to violence.

The VPD, funded by the {\color{black}US} National Institute of Justice, is free {\color{black}and} public{\color{black}ly} available {\color{black}at \texttt{theviolenceproject.org}. The database, which provides data on} mass shootings in the US from 1966 to the present day, {\color{black}is} widely regarded as the largest and most comprehensive database of mass shooters to date, {\color{black}collecting} more than 150 variables {\color{black}that include} demographic characteristics of victims and offenders, life history information and mental health indicators of the perpetrators, and the types of firearms used in the shootings. The data are primarily collected from first-person accounts such as diaries, suicide notes, social media and blog posts, audio and video recordings, interview transcripts, and personal correspondence with the perpetrators. In addition, information is gathered also using media coverage (television, newspapers, magazines){\color{black},} documentaries{\color{black},} biographies{\color{black},} books and articles{\color{black},} court transcripts{\color{black},} law enforcement records{\color{black},} medical records{\color{black},} school records{\color{black},} and autopsy reports.

The dataset considered in {\color{black}our analysis was} restricted to $188$ mass shootings occurred from 1966 to 2022 {\color{black}that were} committed by male perpetrators {\color{black}only} ({\color{black}this is because, overall,} 97\% of shooters {\color{black}in the database} are men). Due to incomplete records, we remove{\color{black}d} the Wilkinsburg shooting occurred on March 9, 2016. In order to understand the complexities beneath MPSs, we considered the following variables. For each attack, we include{\color{black}d} the number of killed and injured people, the number of firearms brought to the scene and the median age {\color{black}(years)} of the victims. Concerning the offender personal characteristics we consider{\color{black}ed} age {\color{black}(years)}, two binary variables controlling for immigrant {\color{black}status} {\color{black}(reference: no)} and the relationship {\color{black}(including marriage)} status {\color{black}(reference: not in a relationship)}. The analysis also {\color{black}took} into account perpetrator specific traits linked to past personal experiences {\color{black}that span} six domains: social contagion (7 items), crime and violence (14 items), traumas and adverse childhood (15 items), crisis (11 items), health and mental health (14 items) and grievance/motivation (14 items). In particular, the items {\color{black}in each of} the six domain {\color{black}were} first {\color{black}dichotomized based on} whether there was evidence that the symptoms in question had been noted prior or during the shooting. {\color{black}Subsequently, these domain-specific items were averaged to give a summary of each of those six domains}. Finally, {\color{black}we included} a binary variable {\color{black}indicating whether the shooter had a prior relationship with the shooting site (reference: no).} {\color{black}This gave 14 variables in all and are described in Table \ref{tab:desc}}.

\begin{table}[h!]
\centering
\resizebox{0.8\columnwidth}{!}{%
\begin{tabular}{l|l}
  \toprule
Variable & Description\\
\hline
Killed & Number of fatalities\\
Injured & Number of people injured\\
Firearms brought to the scene & Total weapons brought to the scene\\
Age & Age in years of the perpetrator\\
Victims age & Median age in years of the victims per attack\\
Insider & Shooter has an existing relationship with the shooting site \\
Immigrant & Shooter is an immigrant \\
Relationship status & Shooter is single/divorced/separated/widowed or in a relationship/married \\
  Social & Social contagion\\
  \qquad Leakage & \\
  \qquad Interest in Past Mass Violence \\
  \qquad Relationship with Other Shooting(s) \\
  \qquad Legacy Token & \\
  \qquad Pop Culture Connection & \\
  \qquad Planning & \\
  \qquad Performance & \\
  Crime & Crime and violence\\
  \qquad Known to Police or FBI & \\
  \qquad Criminal Record & \\
  \qquad Part I Crimes & \\
  \qquad Part II Crimes & \\
  \qquad Highest Level of Justice System Involvement & \\
  \qquad Suspected/Arrested/Charged/Convicted & \\
  \qquad History of Physical Altercations & \\
  \qquad History of Animal Abuse & \\
  \qquad History of Domestic Abuse & \\
  \qquad History of Sexual Offenses & \\
  \qquad Gang Affiliation & \\
  \qquad Terror Group Affiliation & \\
  \qquad Known Hate Group or Chat Room Affiliation & \\
  \qquad Bully & \\
  Traumas & Trauma and adverse childhood\\
  \qquad Bullied & \\
  \qquad Raised by Single Parent & \\
  \qquad Parental Divorce/Separation & \\
  \qquad Parental Death in Childhood & \\
  \qquad Parental Suicide & \\
  \qquad Childhood Trauma & \\
  \qquad Physically Abused & \\
  \qquad Sexually Abused & \\
  \qquad Emotionally Abused & \\
  \qquad Neglected & \\
  \qquad Mother Violent Treatment & \\
  \qquad Parental Substance Abuse & \\
  \qquad Parent Criminal Record & \\
  \qquad Family Member Incarcerated & \\
  \qquad Adult Trauma & \\
  Crisis & Sign of a crisis\\
  \qquad Recent or Ongoing Stressor & \\
  \qquad Signs of Being in Crisis & \\
  \qquad Inability to Perform Daily Tasks & \\
  \qquad Notably Depressed Mood & \\
  \qquad Unusually Calm or Happy & \\
  \qquad Rapid Mood Swings & \\
  \qquad Increased Agitation & \\
  \qquad Abusive Behavior & \\
  \qquad Isolation & \\
  \qquad Losing Touch with Reality & \\
  \qquad Paranoia & \\
  Mental & Health and mental health\\
  \qquad Suicidality & \\
  \qquad Prior Hospitalization & \\
  \qquad Voluntary or Involuntary Hospitalization & \\
  \qquad Prior Counseling & \\
  \qquad Voluntary or Mandatory Counseling & \\
  \qquad Psychiatric Medication & \\
  \qquad Treatment 6 Months Prior to Shooting & \\
  \qquad Mental Illness & \\
  \qquad Fetal Alcohol Spectrum Disorder & \\
  \qquad Known Family Mental Health History & \\
  \qquad Autism Spectrum & \\
  \qquad Substance Use & \\
  \qquad Health Issues & \\
  \qquad Head Injury/Possible TBI & \\
  Motivation & Grievance and motivation\\
  \qquad Known Prejudices & \\
  \qquad Racism/Xenophobia & \\
  \qquad Religious Hate & \\
  \qquad Misogyny & \\
  \qquad Homophobia & \\
  \qquad Employment Issue & \\
  \qquad Economic Issue & \\
  \qquad Legal Issue & \\
  \qquad Relationship Issue & \\
  \qquad Interpersonal Conflict & \\
  \qquad Fame-Seeking & \\
  \qquad Other & \\
  \qquad Unknown & \\
  \qquad Role of Psychosis in the Shooting & \\
 \bottomrule
\end{tabular}%
}
\caption{Description of the variables in the sample.}
\label{tab:desc}
\end{table}

\begin{table}[h!]
\centering
\resizebox{1.0\columnwidth}{!}{%
\begin{tabular}{lrrrrrr}
  \toprule
Variable & Minimum & First quartile & Mean & Median & Third quartile & Maximum \\
  \hline
  Killed & 4.00 & 4.00 & 7.22 & 5.00 & 7.00 & 60.00 \\
  Injured & 0.00 & 1.00 & 11.34 & 3.00 & 7.00 & 867.00 \\
  Firearms brought to the scene & 1.00 & 1.00 & 2.23 & 2.00 & 3.00 & 24.00 \\
  Age & 11.00 & 23.00 & 33.61 & 32.00 & 43.00 & 70.00 \\
  Victims age & 6.00 & 28.00 & 38.72 & 39.25 & 49.00 & 88.00 \\
  Social & 0.00 & 0.00 & 0.23 & 0.14 & 0.29 & 1.00 \\
  Crime & 0.00 & 0.14 & 0.28 & 0.29 & 0.43 & 0.79 \\
  Traumas & 0.00 & 0.00 & 0.12 & 0.07 & 0.20 & 0.73 \\
  Crisis & 0.00 & 0.27 & 0.41 & 0.46 & 0.55 & 0.82 \\
  Mental & 0.00 & 0.14 & 0.27 & 0.21 & 0.43 & 0.79 \\
  Motivation & 0.07 & 0.12 & 0.16 & 0.14 & 0.21 & 0.36 \\
  \\
 & Frequency & Proportion (\%) \\
  \hline
  Insider (no) & 82 & 43.62 \\
  Immigrant (no) & 28 & 14.89 \\
  Relationship status (single) & 57 & 30.32 \\
  \bottomrule
\end{tabular}%
}
\caption{Summary statistics of the variables in the sample. For binary variables we report the absolute frequency, proportion in percentage and the baseline category in parenthesis.}
\label{tab:summary}
\end{table}

\section{Methods}\label{sec:met}
In this section we illustrate the proposed quantile mixed graphical model {\color{black}(QMGM)}. {\color{black}First,} we extend the mid-quantile regression of \cite{geraci2022mid} to the graphical modeling framework with both continuous and discrete variables. Subsequently, using the neighborhood selection approach of \cite{meinshausen2006high}, we show how to estimate a sparse mixed graphical model characterizing conditional independence relations among variables via node-wise penalized mid-quantile regressions.\\ 

Let $\bs Y = (X_1, \dots, X_{p_1}, Z_1, \dots, Z_{p_2})'$ denote a $p$-dimensional random vector, where $X_1, \dots, X_{p_1}$ are $p_1$ {\color{black}absolutely} continuous variables and $Z_1, \dots, Z_{p_2}$ are $p_2$ discrete variables. {\color{black} In particular, $Y_j$, for $j=1,\dots,p$, can be continuous, binary, ordinal, or count, with positive or negative values, not necessarily equally spaced {\color{black}and not necessarily} integers. However, as in \cite{geraci2022mid}, we exclude variables with {\color{black}a nominal scale}.} Also, let $\mathcal{G} = (V, E)$ denote an undirected graph where $V = \{1,\dots,p\}$ is the set of nodes such that each component of the random variable $\bs Y$ corresponds to a node in $V$, and $E \subseteq V \times V$ represents the set of undirected edges. Following \cite{geraci2022mid}, we {\color{black}define} the conditional mid-cumulative distribution function (mid-CDF, \citealp{parzen1993change, parzen2004quantile}) of $Y_j$ given all other variables as
\begin{equation}\label{eq:midCDF}
G_{Y_j \mid \bs Y_{\neg j}} (y_j \mid \bs y_{\neg j}) = F_{Y_j \mid \bs Y_{\neg j}} (y_j \mid \bs y_{\neg j}) - 0.5 m_{Y_j \mid \bs Y_{\neg j}} (y_j \mid \bs y_{\neg j}),
\end{equation}
where $\bs Y_{\neg j}$ denotes all variables except $Y_j$, $F_{Y_j \mid \bs Y_{\neg j}} (\cdot \mid \cdot)$ is the conditional CDF of $Y_j$ {\color{black}given $\bs Y_{\neg j}$} and $m_{Y_j \mid \bs Y_{\neg j}} (y_j \mid \bs y_{\neg j}) = \mbox{Pr} (Y_j = y_j \mid\bs Y_{\neg j} = \bs y_{\neg j})$. The definition of conditional mid-CDF in \eqref{eq:midCDF} applies to both continuous and discrete variables. Indeed, if $Y_j$ is discrete, $G_{Y_j \mid \bs Y_{\neg j}} (y_j \mid \bs y_{\neg j})$ is a step function {\color{black}(shifted downwards relative to $F$)} while it reduces to $F_{Y_j \mid \bs Y_{\neg j}} (y_j \mid \bs y_{\neg j})$ if $Y_j$ is continuous since $\mbox{Pr} (Y_j = y_j \mid\bs Y_{\neg j} = \bs y_{\neg j}) = 0$.

Let $\mathcal{S}_{Y_j}$ be the set of $s$ distinct values in the population that the random variable $Y_j$ can take on. In particular, $\mathcal{S}_{Y_j}$ can be a finite or a countably infinite $(s = \infty)$ subset of $\mathbb{R}$.

Then, the conditional mid-quantile function (mid-QF) of $Y_j$, $H_{Y_j \mid \bs Y_{\neg j}} (\tau)$, is defined as the piecewise linear function connecting the values $G^{-1}_{Y_j \mid \bs Y_{\neg j}} (\pi_{jh} \mid \bs y_{\neg j})$, where $\pi_{jh} = G_{Y_j \mid \bs Y_{\neg j}} (y_j \mid \bs y_{\neg j})$, $h = 1,\dots,s$, for a given quantile level $\tau \in (0,1)$. We model the $\tau$-th conditional mid-quantile of $Y_j$ given all the other variables {\color{black}with}
\begin{equation}\label{eq:model}
H_{g_j (Y_j) \mid \bs Y_{\neg j}} (\tau) = \beta_j^0(\tau) + \bs y'_{\neg j} \boldsymbol \beta_j (\tau), \quad j=1,\dots,p,
\end{equation}
where $g_j (\cdot)$ is a known monotone and differentiable \qmo link\qmcsp function, and $\boldsymbol \beta_j(\tau) = (\beta_j^1(\tau), \dots, \beta_j^{p-1}(\tau))'$ is a vector of $p-1$ unknown regression coefficients, with $\beta_j^0(\tau)$ being an intercept term, for a given $\tau$. In our approach, $g_j (\cdot)$ may simply be the identity or logarithmic function, which is typically used in the modeling of counts, or the logistic function, but other more flexible transformations can also be employed \citep{Mu2007, geraci2015improved}.

To study conditional independence relations between the components of $\bs Y$ through the graph $\mathcal{G}$, we establish a result that allows us to make inference on the edge structure $E$ using mid-quantile regressions. Following \cite{ali2016multiple} and \cite{chun2016graphical}, the next proposition characterizes the relationship between the conditional mid-quantile function in \eqref{eq:model} and the conditional independence between any pair of variables in $\bs Y$ given the rest.
\begin{proposition}\label{prop:indep}
Suppose that the conditional mid-QF of a random variable $Y_j$, for some $j=1,\dots,p$, is defined by the mid-quantile regression model in \eqref{eq:model}. Then, $Y_j$ is conditionally independent from $Y_k$, with $k=1,\dots,p$ and $k \neq j$, given all of the other variables if and only if $\beta_j^k(\tau) = 0$ for all $\tau \in (0,1)$.
\end{proposition}
\begin{proof}
If the random variable $Y_j$ is {\color{black}absolutely} continuous, then the result follows from the fact that the mid-CDF in \eqref{eq:midCDF} reduces to the conditional CDF, $F_{Y_j \mid \bs Y_{\neg j}} (y_j \mid \bs y_{\neg j})$, and the conditional mid-quantile reduces to the standard conditional quantile of $Y_j$ given all other variables. In this case, if the conditional quantiles satisfy $H_{g_j (Y_j) \mid \bs Y_{\neg j}} (\tau) = H_{g_j (Y_j) \mid \bs Y_{\neg {\color{black}\{j,k\}}}} (\tau)$ for all $\tau \in (0,1)$, then the conditional CDF must obey the same property, i.e., $F_{Y_j \mid \bs Y_{\neg j}} (y_j \mid \bs y_{\neg j}) = F_{Y_j \mid \bs Y_{\neg \{j,k\}}} (y_j \mid \bs y_{\neg \{j,k\}})$. The converse of this statement is true as well by reversing all the arguments.

When $Y_j$ is discrete, if $\beta_j^k(\tau) = 0$ for all $\tau \in (0,1)$ then the conditional mid-quantiles of $Y_j$ do not depend on $Y_k$, i.e., $H_{g_j (Y_j) \mid \bs Y_{\neg j}} (\tau) = H_{g_j (Y_j) \mid \bs Y_{\neg \{j,k\}}} (\tau)$. Since the mid-CDF $G_{Y_j \mid \bs Y_{\neg j}} (y_j \mid \bs y_{\neg j})$ is purely determined by $H_{g_j (Y_j) \mid \bs Y_{\neg j}} (\tau)$, we have that
\begin{equation*}
G_{Y_j \mid \bs Y_{\neg j}} (y_j \mid \bs y_{\neg j}) = F_{Y_j \mid \bs Y_{\neg j}} (y_j \mid \bs y_{\neg j}) - 0.5 m_{Y_j \mid \bs Y_{\neg j}} (y_j \mid \bs y_{\neg j})
\end{equation*}
and the left-hand side does not depend on $Y_k$, so neither can the right-hand side. This implies that the distribution $F_{Y_j \mid \bs Y_{\neg j}} (y_j \mid \bs y_{\neg j})$ equals $F_{Y_j \mid \bs Y_{\neg \{j,k\}}} (y_j \mid \bs y_{\neg \{j,k\}})$, i.e., $Y_j$ and $Y_k$ are conditionally independent given all of the other variables. To complete the proof we note that the converse of the statement is true as well by reversing the arguments.
\end{proof}

The proof of Proposition \ref{prop:indep} follows from the relationship between the conditional mid-quantile and the mid-CDF of each node given the others. Most importantly, from Proposition \ref{prop:indep} it follows that the elements of the vector $\boldsymbol \beta_j$ that are zero for all $\tau \in (0,1)$ correspond to conditional independence relations between the components of $\bs Y$. Hence, the edge set $E$ of the graph $\mathcal{G}$ is completely determined by the non-zero components in $\boldsymbol \beta_j (\tau)$, that is, $(j, k) \in E$ if and only if $\beta_j^k(\tau) \neq 0$. Based on this result, we can build a mixed quantile graphical model to characterize conditional independence relationships between the elements of $\bs Y$ by inferring the zero elements in $\boldsymbol \beta_j (\tau)$, $j=1,\dots,p$.

We exploit the neighborhood selection approach of \cite{meinshausen2006high} by running separate mid-quantile regressions of each component in $\bs Y$ on all the others. Specifically, let $\boldsymbol \tau = (\tau_1, \dots, \tau_L)$ be a grid of $L$ ordered quantile levels with $\tau_l \in (0,1)$, $l=1,\dots,L$. Large values of $L$ allow us to investigate conditional independence more accurately, but they also increase the computational cost of estimating the model. To infer the graph structure, we consider the linear model in \eqref{eq:model} for the conditional mid-QF, $H_{g_j (Y_j) \mid \bs Y_{\neg j}} (\tau_l)$, over all variables $j=1,\dots,p$ and levels $l=1,\dots,L$. Consequently, the corresponding edge set $E$ of conditional dependencies is defined as
\begin{equation}\label{eq:edgeset}
E = \Big\{ (j,k): \underset{l=1,\dots,L}{\max} \{ \max \{ \mid \beta_j^k (\tau_l) \mid, \mid \beta_k^j (\tau_l) \mid \} \} > 0, \quad \textnormal{for} \quad 1 \leq j \neq k \leq p \Big\}.
\end{equation}
With respect to existing approaches in the literature, if all variables are continuous our model reduces to the sparse quantile-based graphical model of \cite{chun2016graphical} and the multiple quantile graphical model of \cite{ali2016multiple} when we estimate the $\tau_l$-th conditional quantile of each variable $Y_j$ given the other variables, using LASSO penalized linear regressions. 

In the next section, we describe a procedure to estimate the proposed graphical model $\mathcal{G}$ and induce sparsity in the regression coefficients.

\section{Estimation}\label{sec:est}
Consider a sample $\bs Y_i, i = 1,\dots,n$, with corresponding observations $\bs y_i$. For each {\color{black}component of $\bs Y_{i}$} and level $\tau_l$, $l=1,\dots,L$, estimation of the model in \eqref{eq:model}, and in turn, of the set $E$ in \eqref{eq:edgeset}, proceeds in two steps.

Let $z_{jh}$, $h=1,\dots,k$, be the {\color{black}$h$th} distinct observation of $Y_j$ that occurs in the sample, with $z_{jh} < z_{jh+1}$ for all $h=1,\dots,k-1$. In the first step we estimate the mid-CDF in \eqref{eq:midCDF}, $\widehat{G}_{Y_j \mid {\bs Y}_{\neg j}} (y_j \mid {\bs y}_{\neg j})$, where $\widehat{F}_{Y_j \mid {\bs Y}_{\neg j}}$ is obtained by fitting $k$ separate logistic regressions, one for each value of $z_{jh}$, $h=1,\dots,k$, and $\widehat{m}_{Y_j \mid \bs Y_{\neg j}} (z_{jh} \mid \bs y_{\neg j}) = \widehat{F}_{Y_j \mid {\bs Y}_{\neg j}} (z_{jh} \mid \bs y_{\neg j}) - \widehat{F}_{Y_j \mid {\bs Y}_{\neg j}} (z_{jh-1} \mid \bs y_{\neg j})$. This idea was originally considered by \cite{foresi1995conditional} and \cite{peracchi2002estimating} to address the curse of dimensionality of non-parametric estimators and, while originally a logit estimator was proposed, in principle any other link function can be employed. To ensure that the CDF is monotonic, we then monotonize the estimates $\widehat{F}_{Y_j \mid {\bs Y}_{\neg j}}$ by rearrangement \citep{chernozhukov2010quantile}. Alternatively, if $p$ is not large, we also note that the conditional CDF can be non-parametrically estimated using the kernel estimator proposed by \cite{li2008nonparametric}.

In the second step, we define $\widehat{G}^c_{Y_j \mid {\bs Y}_{\neg j}} (y_j \mid {\bs y}_{\neg j})$ as the function interpolating the points $(z_{jh} , \widehat{G}_{Y_j \mid {\bs Y}_{\neg j}} (z_{jh} \mid {\bs y}_{\neg j}))$, where the ordinates have been obtained in the first step. The goal now is to estimate $(\beta^0(\tau_l), \boldsymbol \beta_j(\tau_l))$ in \eqref{eq:model} by solving the implicit equation $\tau_l = \widehat{G}^c_{Y_j \mid {\bs Y}_{\neg j}} (\eta(\tau_l) \mid {\bs y}_{\neg j})$, where $\eta(\tau_l) = g^{-1}_j \{ \beta^0(\tau_l) + {\bs y}'_{\neg j} \boldsymbol \beta_j(\tau_l) \}$. {\color{black}Thus far, the proposed estimation procedure follows \cite{geraci2022mid}. However, since our goal is to capture the most relevant relationships between the variables, we extend \citeauthor{geraci2022mid}'s (\citeyear{geraci2022mid}) objective function by adding a LASSO type penalty on $\boldsymbol \beta_j(\tau_l)$, which results in the following estimator}
\begin{equation}\label{eq:obj}
{\color{black} \widehat{\boldsymbol{\beta}}_j(\tau_l)} = \underset{\boldsymbol \beta}{\argmin} \, \frac{1}{n} \sum_{i=1}^n \Big( \tau_l - \widehat{G}^c_{Y_j \mid {\bs Y}_{\neg j}} (\eta_i \mid {\bs y}_{\neg j}) \Big)^2 + \lambda \mid \mid \mbox{diag}(\bs w) \boldsymbol \beta_j(\tau_l) \mid \mid_1,
\end{equation}
where
\begin{equation}
\widehat{G}^c_{Y_j \mid {\bs Y}_{\neg j}} (\eta_i \mid {\bs y}_{\neg j}) = b_{h_i} (\eta_i - z_{jh_i}) + \widehat{\pi}_{jh_i}, \quad z_{jh_i} \leq \eta_i \leq z_{jh_i+1},
\end{equation}
{\color{black}is an interpolation function,} with $b_{h_i} = \frac{\widehat{\pi}_{jh_i+1} - \widehat{\pi}_{jh_i}}{z_{jh_i+1} - z_{jh_i}}$ and $\widehat{\pi}_{jh_i} = \widehat{G}_{Y_j \mid {\bs Y}_{\neg j}} (z_{jh_i} \mid {\bs y}_{\neg j})$. The penalization in \eqref{eq:obj} {\color{black}allows for} a different weight for each coefficient by using the vector $\bs w$ to avoid that variables of different types are on different scales, and where $\lambda \geq 0$ is the overall tuning parameter of the model. The parameter $\lambda$ controls the strength of the penalization and determines the sparsity of the graph: a higher (lower) value is responsible for a lower (higher) number of edges; when $\lambda = 0$, $\widehat{\boldsymbol \beta}_j(\tau_l)$ reduces to the closed-form estimator in \citet[][eq. 2.9]{geraci2022mid}. Finally, to infer the graph structure we solve the minimization problem in \eqref{eq:obj} for all $Y_j$, $j=1,\dots,p$ and $\tau_l$, $l=1,\dots,L$, and estimate the edge set $E$ as follows:
\begin{equation}\label{eq:edgeset_est}
\widehat{E} = \Big\{ (j,k): \underset{l=1,\dots,L}{\max} \{ \max \{ \mid \widehat{\beta_j^k} (\tau_l) \mid, \mid \widehat{\beta_k^j} (\tau_l) \mid \} \} > 0, \quad \textnormal{for} \quad 1 \leq j \neq k \leq p \Big\}.
\end{equation}

To select the optimal value of the penalty parameter $\lambda$, we adopt{\color{black}ed} the following {\color{black}Bayesian Information Criterion (BIC)}:
\begin{equation}\label{eq:BIC}
\mbox{BIC} (\lambda) = \sum_{l=1}^L \sum_{j=1}^p \Bigg[ \ln \bigg( \sum_{i=1}^n \rho_\tau (y_{ij} - \beta_j^0(\tau_l) - \bs y'_{i\neg j} \boldsymbol \beta_j (\tau_l) ) \bigg) + \nu_{jl} \frac{\ln n \ln (p-1)}{2n} C_n \Bigg],
\end{equation}
where $\rho_\tau (u) = u (\tau - I(u < 0))$ is the quantile loss function \cite{koenker1978regression}, with $I(\cdot)$ being the indicator function, $\nu_{jl}$ is the number of estimated non-zero components in $\widehat{\boldsymbol \beta}_j(\tau_l)$ for node $j$ at quantile level $\tau_l$ and $C_n$ is some positive constant, which diverges to infinity as $n$ increases. Specifically, we fit the model for a grid of candidate values of $\lambda$ and then select the optimal tuning parameter as that corresponding to the lowest BIC value in \eqref{eq:BIC}.

\section{Simulation {\color{black}studies}}\label{sec:sim}
In this section, we illustrate the performance of the proposed {\color{black}methods using} simulated data. We consider a mixed network of $p = 10$ nodes containing $p_1 = p/2 = 5$ continuous variables and $p_2 = p/2 = 5$ discrete variables. Following \cite{chun2016graphical}, the graph {\color{black}was} generated from the following conditional models:
\begin{align*}
F^{-1}_{Y_1} (u_1 \mid {\bs Y}_{\neg 1} &= {\bs y}_{\neg 1}) = F^{-1}_{\mathcal{T}_3} (u_1) \\
F^{-1}_{Y_2} (u_2 \mid {\bs Y}_{\neg 2} &= {\bs y}_{\neg 2}) = - 0.5 u^2_2 (y_1 + 3) \\
F^{-1}_{Y_3} (u_3 \mid {\bs Y}_{\neg 3} &= {\bs y}_{\neg 3}) = y_1 + F^{-1}_{\textit{Gamma} (\sigma_3, 2)} (u_3), \quad \sigma_3 = \mid y_1 \mid + 0.1 \\
F^{-1}_{Y_4} (u_4 \mid {\bs Y}_{\neg 4} &= {\bs y}_{\neg 4}) = 0.1 (y_3 + 5)^2 F^{-1}_{\mathcal{N} (0, \sigma^2_4)} (u_4), \quad \sigma_4 = \sqrt{ \mid y_3 + 5 \mid} \\
F^{-1}_{Y_5} (u_5 \mid {\bs Y}_{\neg 5} &= {\bs y}_{\neg 5}) = 2 \cos (\pi y_1 / 4) (u_5 - 0.5) (y_1 + 2) + F^{-1}_{\mathcal{N} (0, \sigma^2_5)} (u_5), \quad \sigma_5 = 0.1 + 0.1 \mid y_1 \mid \\
F^{-1}_{Y_6} (u_6 \mid {\bs Y}_{\neg 6} &= {\bs y}_{\neg 6}) = \lfloor (u_6 + 0.5) \mid y_1 \mid \rfloor + DU(1, 3)\\
F^{-1}_{Y_7} (u_7 \mid {\bs Y}_{\neg 7} &= {\bs y}_{\neg 7}) = F^{-1}_{Pois} (u_7, \mid y_3 + 5 \mid^{-1/2} + \mid \log(\mid y_5 \mid + 1) \mid ) \\
F^{-1}_{Y_8} (u_8 \mid {\bs Y}_{\neg 8} &= {\bs y}_{\neg 8}) = \lfloor u_8 y_7 + \mid y_2 + 0.5 \mid^{1.3} \rfloor + DU(1, 3) \lfloor 1 + \mid y_5 \mid \rfloor \\
F^{-1}_{Y_9} (u_9 \mid {\bs Y}_{\neg 9} &= {\bs y}_{\neg 9}) = \lfloor 1 + u_9 y_8 \rfloor + DU(1, 5)\\
F^{-1}_{Y_{10}} (u_{10} \mid {\bs Y}_{\neg 10} &= {\bs y}_{\neg 10}) = F^{-1}_{Pois} (u_{10}, \exp (0.8 u_{10} \log (\mid y_9 + 0.1 \mid))),
\end{align*}
where $u_1, \dots, u_{p}$ {\color{black}were} independently drawn from continuous uniform distributions on $(0, 1)$. {\color{black}Moreover, $DU(a,b)$ denotes a random variable with discrete uniform distribution on $(a, b)$, while} $F^{-1}_{\mathcal{T}}$, $F^{-1}_{\textit{Gamma}}$, $F^{-1}_{\mathcal{N}}$ and $F^{-1}_{Pois}$ denote the quantile functions of the Student $t$, the Gamma, the Normal and the Poisson distribution, respectively. The generated graph consists of $12$ edges and it is reported in Figure \ref{fig:sim_graphs} (left plot), with continuous variables {\color{black}marked by} circles and discrete ones {\color{black}by} squares. The width of the edges is proportional to the absolute value of the strength of the interaction between each pair of nodes while the edge color reflects the sign of the interaction (green for positive and red for negative). In the right{\color{black}-hand plot}, we {\color{black}show also} the adjacency matrix representing the conditional dependencies.

\begin{figure}[h!]
\begin{center}
\includegraphics[scale=.375]{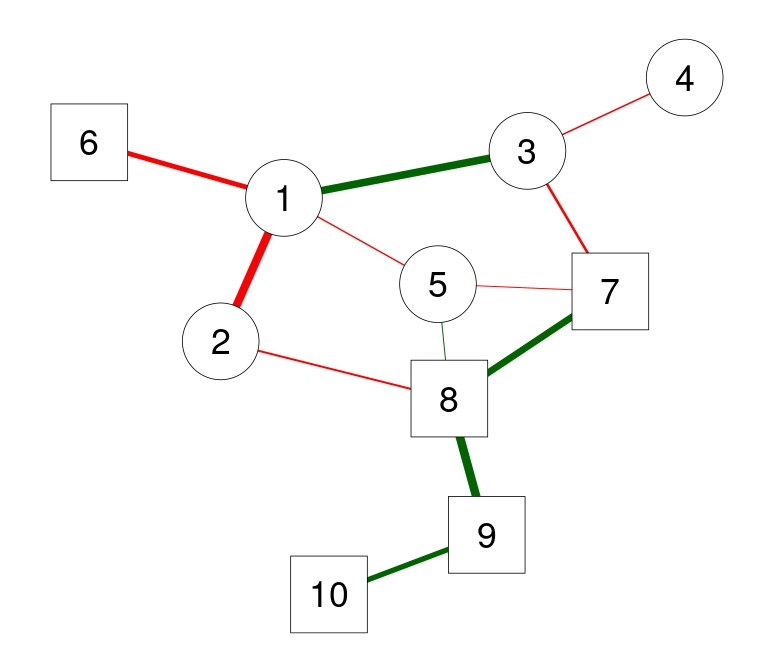}
\includegraphics[scale=.375]{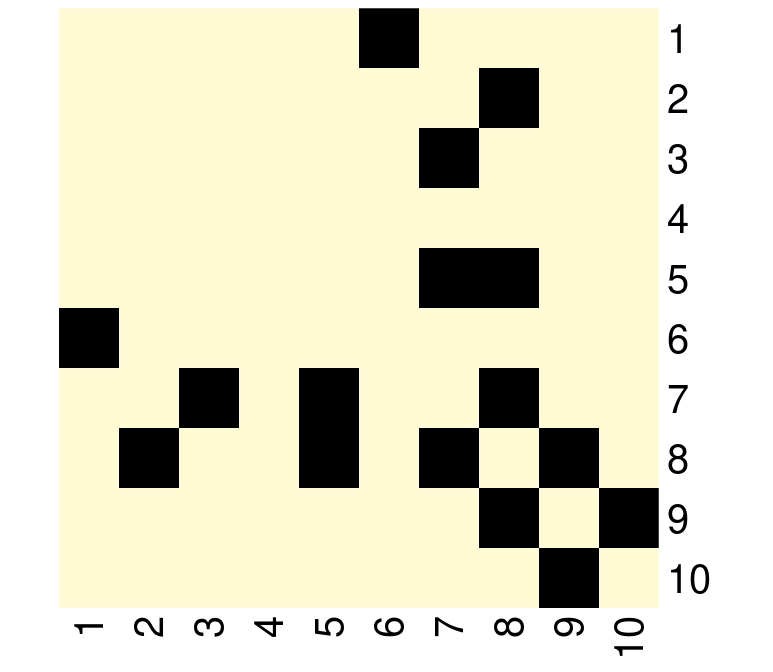}
\caption{Left: {\color{black}graphical representation of the mixed} graph {\color{black}used for the simulation study} with continuous {\color{black}(circles) and discrete (squares)} variables. Right: adjacency matrix {\color{black}that corresponds to the graph on the left}.}
\label{fig:sim_graphs}
\end{center}
\end{figure}

We generated $n \in \{500, 1000\}$ observations {\color{black}from the above-described graph for $R = 100$ replications} and {\color{black}then} fitted the proposed model {\color{black}QMGM}. {\color{black}To investigate the sensitivity of the results to varying number of quantile levels, we considered the following three QMGMs: the median model with $\tau = 0.5$ (QMGM1), the quartile model with $L = 3$ levels $\tau \in \{0.25, 0.5, 0.75\}$ (QMGM3), the octile model with $L = 7$ levels $\tau \in \{0.125, 0.25, \ldots, 0.875\}$ (QMGM7), and a model with $L = 17$ levels $\tau \in \{0.1, 0.15, 0.2, \dots, 0.85, 0.9\}$ (QMGM17). For all models, we used a sequence $\lambda \in \exp(\{\log(0.001), \dots, \log(5)\})$ of 50 equispaced values.} 
Prior to fitting, continuous variables {\color{black}were} centered {\color{black}about} zero and {\color{black}scaled} by their standard deviation. For continuous nodes we {\color{black}took the identity $g(t) = t$} {\color{black}whereas} for discrete nodes we {\color{black}used the log-transform $g(t) = \log(t)$}. The edge set $\widehat{E}$ {\color{black}was} estimated as described in \eqref{eq:edgeset_est}. We compare{\color{black}d} our model with the {\color{black}mean-based} MGM {\color{black}that fits a GLM} on each node with a LASSO penalty. {\color{black}Such a comparison was facilitated by setting} the weight vector $\bs w$ equal to a vector of ones {\color{black} so that the two approaches allow for the same penalty on each coefficient. Finally, we compared five different selection criteria: the BIC penalties $C_n = \log (p-1) = 2.20$ (BICp), $C_n = 2^{-1} \log (p-1) = 1.10$ (BIC2p), $C_n = 3^{-1} \log (p-1) = 0.73$ (BIC3p), and $C_n = 1$ (BIC), as well as the ordinary Akaike Information Criterion (AIC)}. {\color{black}The simulation was carried out in R version 4.3.0 \citep{r2023} with an Intel Xeon E5-2609 2.40 GHz processor, using our own code opportunely adapted from the package \texttt{Qtools} \citep{Qtools} to fit QMGMs and the \texttt{mgm} package \citep{mgm} to fit MGMs.}

As a first {\color{black}aim}, we assess{\color{black}ed} how well {\color{black}each model} recover{\color{black}s} the true edges by reporting {\color{black}the} receiver operating characteristic (ROC) curves, {\color{black}with} the true positive rate (TPR) against the false positive rate (FPR) across the $R$ replicates. {\color{black}The results, averaged over the Monte Carlo replications, are shown in} Figure \ref{fig:sim_ROC}. {\color{black}One can observe that the areas under the ROC curves (AUCs) resulting from our proposed QMGMs are always greater  compared to that of the mean-based MGM, making our proposed method substantially more competitive than the only alternative model currently available for mixed graphs. The performance of QMGM improves with increasing number of quantile levels, thus with a finer coverage of the distribution, although the performance is essentially the same for QMGM7 and QMGM17. This means that the computational burden is no longer justified by the vanishing marginal gain. Summary statistics for the estimated AUCs and computational time for each model are given in Table \ref{tab:sim_AUC}.}

\begin{figure}[h!]
\begin{center}
\includegraphics[scale=.435]{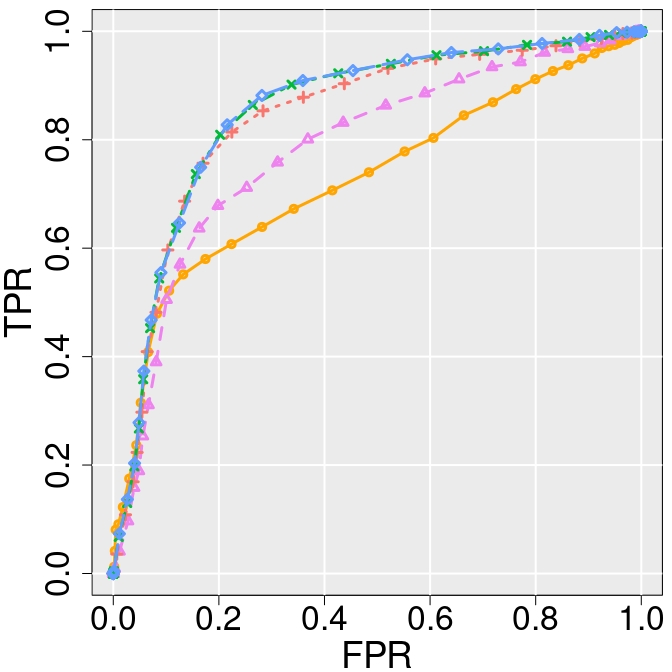}
\includegraphics[scale=.435]{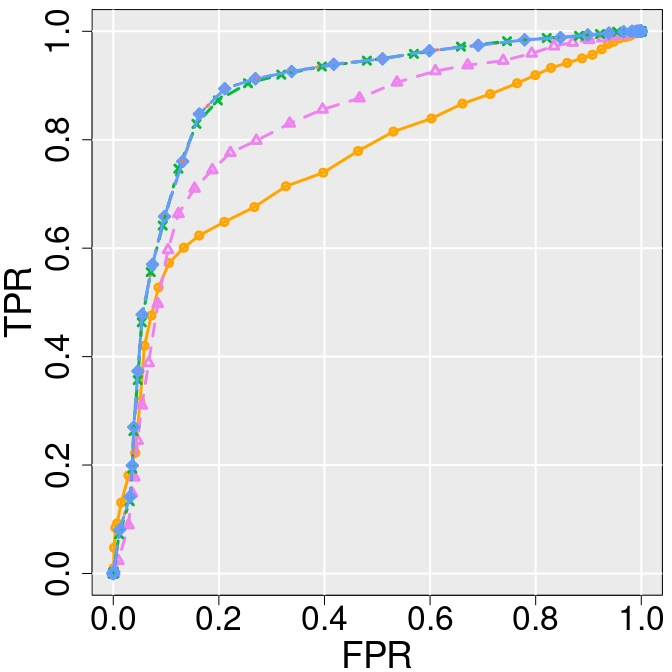}
\caption{{\color{black}Estimated ROC curves averaged over $R = 100$ replication when $n=500$ (left plot) and $n = 1000$ (right plot)} for the MGM (orange), QMGM1 (violet), QMGM3 (red), QMGM7 (green) and QMGM17 (blue).}
\label{fig:sim_ROC}
\end{center}
\end{figure}

\begin{table}
\centering
\resizebox{1\textwidth}{!}{%
\begin{tabular}{l r r r r r}
\toprule
$n$ & MGM & QMGM1 & QMGM3 & QMGM7 & QMGM17 \\
\hline
\multicolumn{6}{l}{Panel A: AUC}\\
$500$ & $0.71 \; [0.63; 0.78]$ & $0.79 \; [0.70; 0.85]$ & $0.84 \; [0.75; 0.88]$ & $0.86 \; [0.76; 0.91]$ & $0.86 \; [0.78; 0.91]$ \\
$1000$ & $0.75 \; [0.69; 0.82]$ & $0.83 \; [0.77; 0.88]$ & $0.88 \; [0.81; 0.92]$ & $0.89 \; [0.83; 0.92]$ & $0.89 \; [0.84; 0.93]$ \\
\\
\multicolumn{6}{l}{Panel B: Computational time {\color{black}(seconds)}}\\
$500$ & $0.92 \; [0.84; 1.18]$ & $4.63 \; [4.33; 5.30]$ & $8.13 \; [7.60; 9.08]$ & $14.65 \; [14.06; 15.70]$ & $30.85 \; [29.82; 32.47]$ \\
$1000$ & $0.97 \; [0.89; 1.64]$ & $107.14 \; [103.47; 113.44]$ & $115.29 \; [111.67; 121.89]$ & $132.37 \; [127.91; 137.88]$ & $173.85 \; [168.92; 179.47]$ \\
\bottomrule
\end{tabular}}
\caption{Median and{\color{black}, in brackets,} 10th and 90th percentile{\color{black}s} {\color{black}of the estimated} AUC values (Panel A) and computational time in seconds (Panel B) averaged over $R = 100$ Monte Carlo replications for the MGM, QMGM1, QMGM3, QMGM7 and QMGM17 models with sample sizes $n \in \{500, 1000\}$.}
\label{tab:sim_AUC}
\end{table}

{\color{black}The second aim of the simulation study was to evaluate the criterion in \eqref{eq:BIC}} to select the optimal sparsity parameter, $\lambda^\star${\color{black}, and thus to recover the true edge structure of the graph.} We {\color{black}therefore} calculate{\color{black}d} the precision, TPR, FPR, F1-score, Matthews correlation coefficient (MCC) and accuracy. {\color{black}In view of the AUC results discussed above, we excluded QMGM3 and QMGM17 from the comparison as they did not add much to the discussion that follows. Table \ref{tab:sim_selec} shows summary statistics for the selected performance measures. In general, significant edges were detected more often and more accurately from our models than from the MGM. However, the AIC was more prone to overfit, resulting in denser graphs at both sample sizes. In contrast, BIC, BIC2p and BIC3p, all achieved similar performance without being significantly affected by the more severe penalty on the complexity of the model. Lastly, BICp worked well in detecting the true non-zero connections overall and resisting against overfit.}

\begin{table}[h]
\centering
\resizebox{1\textwidth}{!}{%
\begin{tabular}{l c c c c c c}
\toprule
 & Precision & TPR & FPR & F1-score & MCC & Accuracy\\
\hline
\multicolumn{7}{l}{Panel A: $n = 500$}\\\\
MGM \\
AIC & $0.34 \; [0.28; 0.42]$ & $0.75 \; [0.67; 0.92]$ & $0.56 \; [0.36; 0.76]$ & $0.47 \; [0.40; 0.55]$ & $0.21 \; [0.05; 0.36]$ & $0.54 \; [0.40; 0.67]$ \\
BIC & $0.38 \; [0.30; 0.50]$ & $0.67 \; [0.58; 0.92]$ & $0.42 \; [0.27; 0.70]$ & $0.48 \; [0.41; 0.58]$ & $0.24 \; [0.09; 0.41]$ & $0.60 \; [0.44; 0.73]$ \\
BICp & $0.47 \; [0.35; 0.62]$ & $0.67 \; [0.50; 0.75]$ & $0.27 \; [0.15; 0.45]$ & $0.54 \; [0.44; 0.62]$ & $0.35 \; [0.17; 0.48]$ & $0.71 \; [0.60; 0.78]$ \\
BIC2p & $0.38 \; [0.30; 0.50]$ & $0.67 \; [0.58; 0.83]$ & $0.39 \; [0.24; 0.67]$ & $0.48 \; [0.41; 0.58]$ & $0.25 \; [0.11; 0.40]$ & $0.62 \; [0.47; 0.73]$ \\
BIC3p & $0.35 \; [0.29; 0.42]$ & $0.75 \; [0.67; 0.92]$ & $0.53 \; [0.33; 0.76]$ & $0.48 \; [0.40; 0.56]$ & $0.24 \; [0.08; 0.37]$ & $0.56 \; [0.40; 0.67]$ \\
QMGM1 \\
AIC & $0.40 \; [0.33; 0.50]$ & $0.83 \; [0.67; 1.00]$ & $0.42 \; [0.27; 0.64]$ & $0.53 \; [0.45; 0.62]$ & $0.34 \; [0.21; 0.47]$ & $0.62 \; [0.51; 0.73]$ \\
BIC & $0.50 \; [0.40; 0.60]$ & $0.75 \; [0.58; 0.83]$ & $0.27 \; [0.18; 0.39]$ & $0.58 \; [0.48; 0.67]$ & $0.40 \; [0.27; 0.53]$ & $0.73 \; [0.64; 0.78]$ \\
BICp & $0.57 \; [0.47; 0.75]$ & $0.67 \; [0.25; 0.83]$ & $0.18 \; [0.03; 0.28]$ & $0.59 \; [0.35; 0.71]$ & $0.43 \; [0.25; 0.60]$ & $0.78 \; [0.71; 0.84]$ \\
BIC2p & $0.50 \; [0.40; 0.60]$ & $0.75 \; [0.58; 0.83]$ & $0.27 \; [0.15; 0.39]$ & $0.58 \; [0.48; 0.67]$ & $0.40 \; [0.27; 0.53]$ & $0.73 \; [0.64; 0.80]$ \\
BIC3p & $0.47 \; [0.39; 0.57]$ & $0.75 \; [0.58; 0.92]$ & $0.30 \; [0.18; 0.48]$ & $0.57 \; [0.50; 0.65]$ & $0.39 \; [0.26; 0.50]$ & $0.71 \; [0.62; 0.78]$ \\
QMGM7 \\
AIC & $0.30 \; [0.27; 0.36]$ & $1.00 \; [0.92; 1.00]$ & $0.85 \; [0.61; 0.97]$ & $0.46 \; [0.42; 0.52]$ & $0.21 \; [0.09; 0.34]$ & $0.38 \; [0.29; 0.54]$ \\
BIC & $0.36 \; [0.31; 0.44]$ & $0.92 \; [0.92; 1.00]$ & $0.61 \; [0.42; 0.73]$ & $0.52 \; [0.46; 0.59]$ & $0.34 \; [0.20; 0.45]$ & $0.53 \; [0.44; 0.67]$ \\
BICp & $0.44 \; [0.38; 0.53]$ & $0.92 \; [0.83; 1.00]$ & $0.42 \; [0.24; 0.55]$ & $0.59 \; [0.49; 0.67]$ & $0.44 \; [0.29; 0.55]$ & $0.67 \; [0.58; 0.76]$ \\
BIC2p & $0.38 \; [0.32; 0.44]$ & $0.92 \; [0.91; 1.00]$ & $0.58 \; [0.42; 0.70]$ & $0.53 \; [0.47; 0.59]$ & $0.34 \; [0.23; 0.45]$ & $0.58 \; [0.47; 0.67]$ \\
BIC3p & $0.34 \; [0.30; 0.41]$ & $1.00 \; [0.92; 1.00]$ & $0.67 \; [0.51; 0.85]$ & $0.50 \; [0.44; 0.57]$ & $0.28 \; [0.16; 0.41]$ & $0.50 \; [0.38; 0.62]$ \\
\hline
\multicolumn{7}{l}{Panel B: $n = 1000$}\\\\
MGM \\
AIC & $0.33 \; [0.28; 0.44]$ & $0.83 \; [0.67; 1.00]$ & $0.62 \; [0.39; 0.85]$ & $0.48 \; [0.41; 0.56]$ & $0.21 \; [0.06; 0.39]$ & $0.50 \; [0.38; 0.67]$ \\
BIC & $0.37 \; [0.31; 0.50]$ & $0.75 \; [0.66; 0.92]$ & $0.47 \; [0.27; 0.70]$ & $0.49 \; [0.43; 0.58]$ & $0.25 \; [0.11; 0.44]$ & $0.58 \; [0.44; 0.73]$ \\
BICp & $0.46 \; [0.35; 0.62]$ & $0.67 \; [0.58; 0.83]$ & $0.30 \; [0.15; 0.48]$ & $0.56 \; [0.46; 0.64]$ & $0.38 \; [0.20; 0.50]$ & $0.70 \; [0.58; 0.80]$ \\
BIC2p & $0.38 \; [0.31; 0.50]$ & $0.75 \; [0.58; 0.92]$ & $0.45 \; [0.24; 0.67]$ & $0.49 \; [0.43; 0.59]$ & $0.26 \; [0.13; 0.44]$ & $0.60 \; [0.47; 0.73]$ \\
BIC3p & $0.33 \; [0.29; 0.44]$ & $0.83 \; [0.67; 0.93]$ & $0.58 \; [0.33; 0.79]$ & $0.48 \; [0.41; 0.56]$ & $0.23 \; [0.08; 0.39]$ & $0.51 \; [0.40; 0.67]$ \\
QMGM1 \\
AIC & $0.47 \; [0.36; 0.60]$ & $0.83 \; [0.55; 0.92]$ & $0.33 \; [0.15; 0.52]$ & $0.59 \; [0.42; 0.67]$ & $0.42 \; [0.27; 0.53]$ & $0.71 \; [0.58; 0.78]$ \\
BIC & $0.53 \; [0.42; 0.65]$ & $0.75 \; [0.25; 0.92]$ & $0.24 \; [0.06; 0.36]$ & $0.62 \; [0.38; 0.71]$ & $0.46 \; [0.30; 0.60]$ & $0.76 \; [0.69; 0.82]$ \\
BICp & $0.60 \; [0.48; 0.73]$ & $0.75 \; [0.25; 0.83]$ & $0.18 \; [0.06; 0.27]$ & $0.64 \; [0.38; 0.75]$ & $0.50 \; [0.30; 0.66]$ & $0.79 \; [0.71; 0.85]$ \\
BIC2p & $0.53 \; [0.43; 0.67]$ & $0.75 \; [0.25; 0.84]$ & $0.24 \; [0.06; 0.33]$ & $0.62 \; [0.38; 0.71]$ & $0.46 \; [0.30; 0.60]$ & $0.76 \; [0.69; 0.82]$ \\
BIC3p & $0.50 \; [0.41; 0.62]$ & $0.83 \; [0.48; 0.92]$ & $0.27 \; [0.12; 0.42]$ & $0.61 \; [0.43; 0.69]$ & $0.44 \; [0.27; 0.57]$ & $0.73 \; [0.66; 0.80]$ \\
QMGM7 \\
AIC & $0.32 \; [0.29; 0.46]$ & $1.00 \; [0.92; 1.00]$ & $0.73 \; [0.39; 0.88]$ & $0.48 \; [0.43; 0.56]$ & $0.27 \; [0.16; 0.43]$ & $0.44 \; [0.36; 0.69]$ \\
BIC & $0.39 \; [0.32; 0.50]$ & $0.92 \; [0.91; 1.00]$ & $0.56 \; [0.24; 0.70]$ & $0.53 \; [0.47; 0.62]$ & $0.37 \; [0.25; 0.50]$ & $0.58 \; [0.47; 0.74]$ \\
BICp & $0.44 \; [0.38; 0.60]$ & $0.92 \; [0.57; 1.00]$ & $0.42 \; [0.12; 0.55]$ & $0.59 \; [0.50; 0.67]$ & $0.44 \; [0.34; 0.56]$ & $0.67 \; [0.58; 0.78]$ \\
BIC2p & $0.39 \; [0.33; 0.52]$ & $0.92 \; [0.83; 1.00]$ & $0.55 \; [0.24; 0.70]$ & $0.55 \; [0.47; 0.62]$ & $0.37 \; [0.25; 0.51]$ & $0.60 \; [0.49; 0.76]$ \\
BIC3p & $0.37 \; [0.32; 0.50]$ & $0.92 \; [0.92; 1.00]$ & $0.61 \; [0.33; 0.76]$ & $0.52 \; [0.47; 0.62]$ & $0.34 \; [0.24; 0.49]$ & $0.56 \; [0.44; 0.73]$ \\
\bottomrule
\end{tabular}}
\caption{Median and{\color{black}, in brackets} 10th and 90th percentile{\color{black}s} of the MGM, QMGM1 and QMGM7 edge recovery performance over $R = 100$ Monte Carlo simulations with sample size $n = 500$ (Panel A) and $n = 1000$ (Panel B).}
\label{tab:sim_selec}
\end{table}

{\color{black}We conducted an additional simulation study to test the sensitivity of our results to the presence of binary variables, which is an extreme form of discreteness when it comes to quantiles}. Specifically, we replace{\color{black}d} the conditional distributions of $Y_7$ and $Y_8$ with the following conditional models:
\begin{align*}
F^{-1}_{Y_7} (u_7 \mid {\bs Y}_{\neg 7} &= {\bs y}_{\neg 7}) = F^{-1}_{Ber} (u_7, 1/\{ 1+\exp [ -2-\mid y_3 + 5 \mid^{-1/2} + \mid \log(\mid y_5 \mid + 1) \mid ] \} ) \\
F^{-1}_{Y_{10}} (u_{10} \mid {\bs Y}_{\neg 10} &= {\bs y}_{\neg 10}) = F^{-1}_{Ber} (u_{10}, 1/\{ \exp [ -3-0.8 u_{10} \log (\mid y_9 + 0.1 \mid) ] \} ),
\end{align*}
where $F^{-1}_{Ber}$ denotes the quantile function of a Bernoulli distribution. For {\color{black}the} binary nodes, we {\color{black}used the logistic transformation $g(t) = \log(t) - \log(1-t)$}. {\color{black}The results, shown in Figure \ref{fig:sim_ROC_bin} for $n=1000$, confirm what we observed in the main simulation study, except that this time the superiority of QMGM1 relative to MGM was not consistent throughout. Note that the mid-median estimator is equivalent to the probability estimator of a binomial regression \citep[see Section 2.2 in][]{geraci2022mid}, therefore the same used in MGM. However, the mid-quantile estimate is obtained by a two-pronged algorithm which increases the estimation variability compared to the GLM estimator and likely explains the loss of performance in a single-quantile QMGM. Hence, we recommend using multiple quantile levels in the presence of binary variables.}

\begin{figure}[h!]
\begin{center}
\includegraphics[scale=.435]{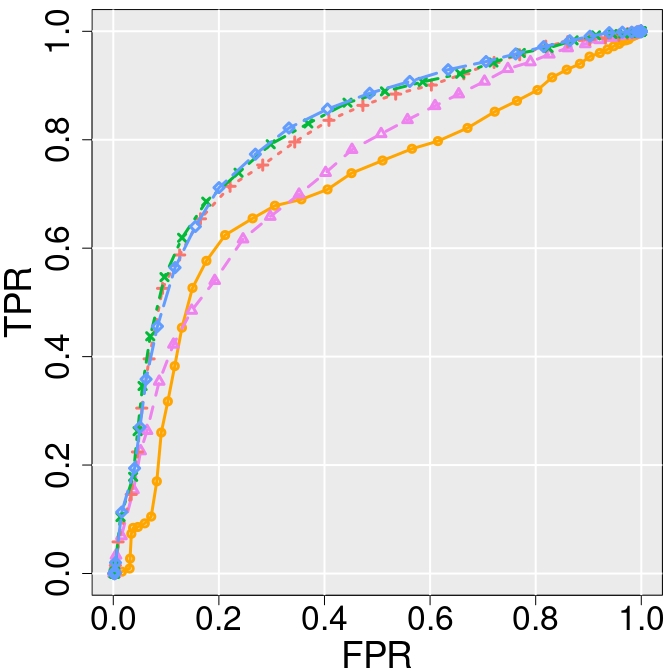}
\caption{{\color{black}Estimated ROC curves averaged over $R = 100$ replications when $n = 1000$ in a scenario with binary variables for the MGM (orange), QMGM1 (violet), QMGM3 (red), QMGM7 (green) and QMGM17 (blue).}}
\label{fig:sim_ROC_bin}
\end{center}
\end{figure}

\section{Application}\label{sec:app}
In this section, we {\color{black}present the results of the analysis of} $n = 188$ MPSs that {\color{black}took} place in the US {\color{black}between} 1966 {\color{black}and} 2022 using the Violence Project data introduced in Section \ref{sec:data}. In the network, we include{\color{black}d the} $p=14$ variables listed in Table \ref{tab:summary}. {\color{black}Before carrying out the analysis, we imputed missing values ($0.65\%$ of the whole dataset) using a $k$-nearest neighbour algorithm where each missing value was substituted by the median of its $k = 13$ closest complete cases.

Figure \ref{fig:dist} depicts the number of fatalities (left) and shootings (right) by year, along with a five year rolling average. The two plots reveal occasional spikes, with an underlying increasing trend for both fatalities and shootings. Descriptive statistics of continuous and discrete variables considered in our analysis are provided in Table \ref{tab:summary}. Perpetrators were aged 33 years on average and 14\% of them had an immigrant status at the time of the shooting. Importantly, many assailants had a connection with the place where the shooting occurred and the majority were not in a relationship/marriage. Many mass shooters were in a noticeable crisis prior to committing the crime and presented some degree of mental health disorder. Moreover, these data suggest that fame-seeking perpetrators are also common, which supports the theory of social contagion among mass killers. As far as the shape of the distributions goes, the number of deaths, injured people and firearms brought to the scene, show significant positive skewness and exhibit outlying values. {\color{black} In this context, a quantile graphical model might prove to be more effective at picking up relevant features of the shootings and of the mass shooters.} 

\begin{figure}[h!]
\begin{center}
\includegraphics[scale=.3]{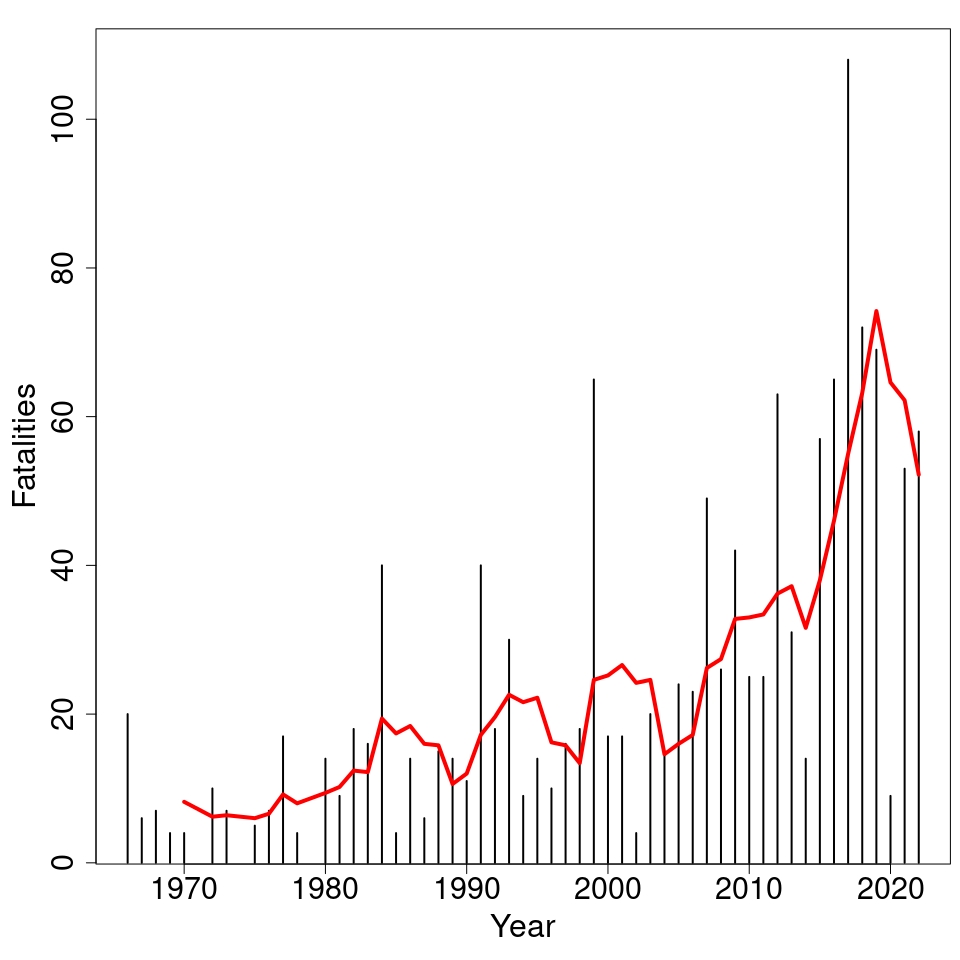}
\includegraphics[scale=.3]{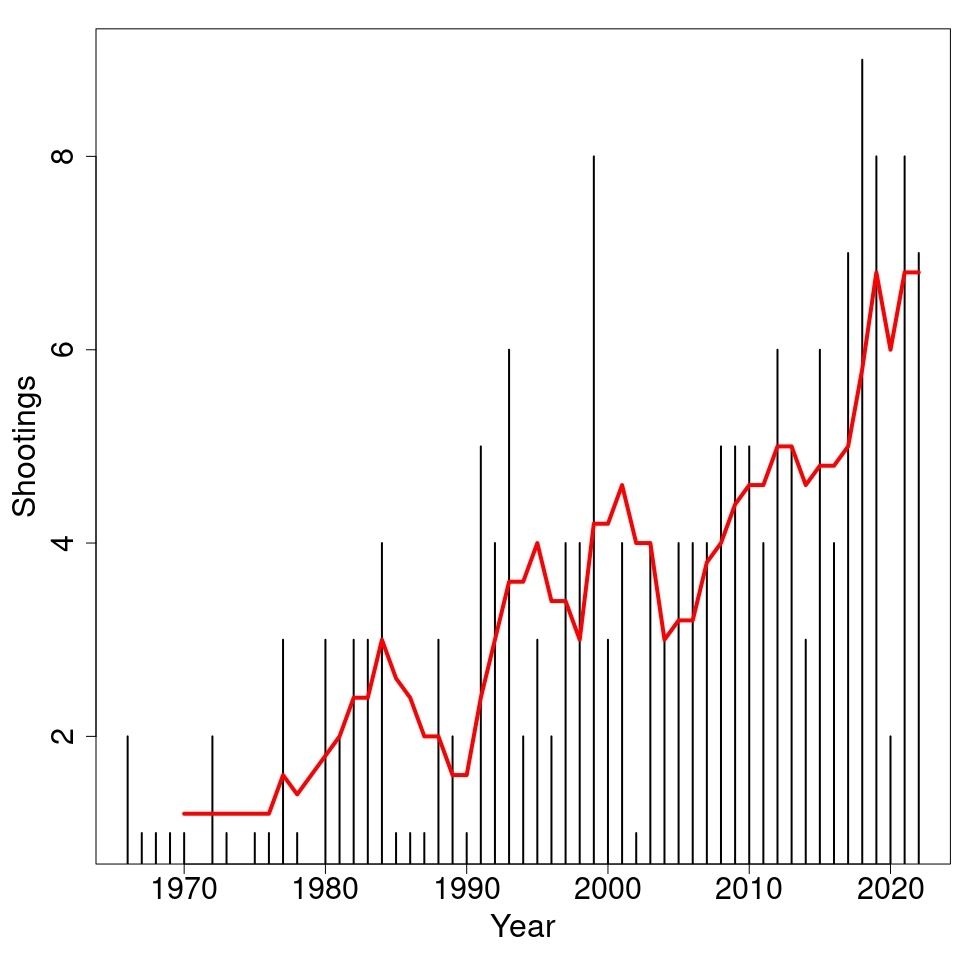}
\caption{Left: number of fatalities per year. Right: number of shootings per year. The red lines denote the five year rolling averages.}
\label{fig:dist}
\end{center}
\end{figure}

Supported by the encouraging results of the simulation study, we therefore applied QMGM to characterize the network of relationships among the selected mass shooting variables and identify those that may better inform prevention strategies.} {\color{black} In particular, we fitted the octile QMGM for a sequence $\lambda \in \exp(\{\log(0.001), \dots, \log(5)\})$ of 100 equispaced values. Prior to fitting, continuous variables were standardized to have zero mean and unit variance. We the identity and logarithmic link functions for continuous and count variables, respectively, and the logistic transformation for binary nodes. The edge set $\widehat{E}$ was estimated as described in \eqref{eq:edgeset_est}. Based on the simulation results, we opted for the standard BIC which corresponds to expression \eqref{eq:BIC} with $C_n = 1$ (we investigated other penalties with $C_n = \log (p-1)$, $2^{-1} \log (p-1)$ and $3^{-1} \log (p-1)$ too, but these produced excessive overfit). We compared our model with the mean-based MGM \citep{yang2014mixed} by fitting LASSO penalized GLMs on each node.}

{\color{black} Figure \ref{fig:graph} provides a representation of the QMGM (top) and MGM (bottom) estimated graphs. The width of the edges is proportional to the absolute value of the strength of the interaction while the edge color reflects the sign of the interaction (green for positive, red for negative, and grey if undefined). The colors of the nodes map to the different domains of the mass shooting variables, namely \emph{shooting} (variables pertaining to the shooting), \emph{characteristics} (socio-demographic characteristics of the shooter and victims) and \emph{background} (variables pertaining to the shooter's background).

We first comment on the network characteristics of the graphs. The graph produced by QMGM appears to be denser than the one produced by MGM with 60 edges estimated by the former and 41 by the latter. The normalized Hamming distance \citep{banks1994metric} between the two graphs, which is a measure of the structural similarities across networks, was equal to 0.363, translating into markedly different graphs. Figure \ref{fig:measures} reports local centrality measures including the degree, betweenness and closeness centrality associated with the two models. QMGM had generally a higher degree and closeness than MGM, which means that the former reveals a stronger interconnection of the graph. The degree measure indicates that the variables age, injured, mental and social are the most connected variables in the network. On the other hand, both betweenness and centrality measures show that the variables injured and social are the most central nodes.
}

{\color{black}As we delve in the specifics of the graph estimated by QMGM (Figure \ref{fig:graph}), we can see that} several shooter characteristics, including the relationship status, personal connection with the place where the shooting took place, crisis and mental health indicators are connected by non-zero edges. We found strong positive conditional dependence relations between past traumas, symptoms of psychosis, mental illness and grievances. Mass shooters seem to share childhood trauma, be{\color{black}ing} in a state of crisis prior to committing their shootings and {\color{black}having committed an act} born out of some motivation/grievance against a specific place {\color{black}or} group of people to blame. In particular, communicating the intent to harm or expressing their motives can be seen both as a call for help and a search for fame and notoriety for their actions. {\color{black}Our} findings are consistent with previous studies \citep{swanson2015mental, metzl2021mental} {\color{black}that} point out that the mental health of the mass shooter is an important factor, although its complexity cannot be easily untangled based on limited information on the mass shooter's psychological profile. 
 The majority of perpetrators {\color{black}had} {\color{black}a pre-}existing, personal relationship to the shooting site {\color{black}(e.g., shooters are often current or former students in school-related shootings, or employees in case of workplaces)}. Among the demographic characteristics, {\color{black}not} being {\color{black}in a relationship was} positively connected with age and the crime, social and mental domains. Shooter's and victims' ages {\color{black}were} positively connected. Though {\color{black}age varies} by shooting location, when {\color{black}a shooter targets} people or a place they know, {\color{black}the victims} tend to be of similar age. The number of {\color{black}victims} (deaths and injured people) and the number of firearms brought to the scene {\color{black}were} positively connected as expected. Most notably, they present non-zero edge{\color{black}s} with the social and motivation {\color{black}variables as well with} the dummy variable \emph{insider}. These {\color{black}findings} likely reflect a {\color{black}number} of concerning behaviors, such as a detailed planning and preparation prior {\color{black}to} the attack, the study of other shootings and fame-seeking motivation \citep{peterson2021communication}.

{\color{black}We conclude with an analysis (results not shown) to assess the sensitivity of the results to different number of quantile levels. Specifically, we considered the QMGMs with $L = 1$ and $L = 17$ values of $\tau$ as in Section \ref{sec:sim}. The number of non-zero edges was 54 for $L=1$ and 62 for $L = 17$, thus indicating that the network density stabilizes with increasing number of $\tau$'s. Moreover, the Hamming distance between the models with $L = 7$ used for the main analysis and the larger model with $L = 17$ was about 0.02, which further supports the robustness of the conclusions.}

\begin{figure}[h!]
\begin{center}
\centering
\includegraphics[scale=.15]{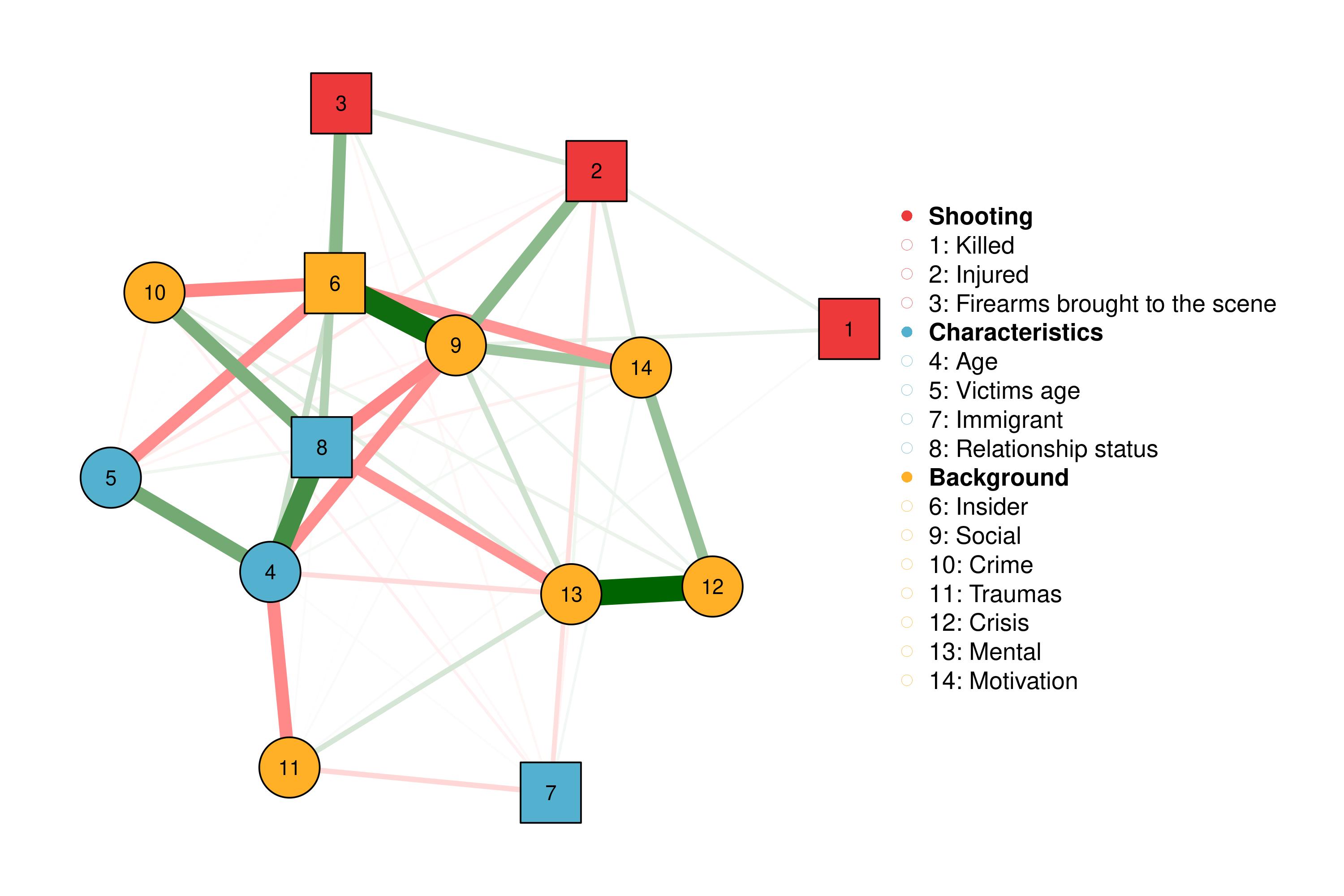}\\
\includegraphics[scale=.15]{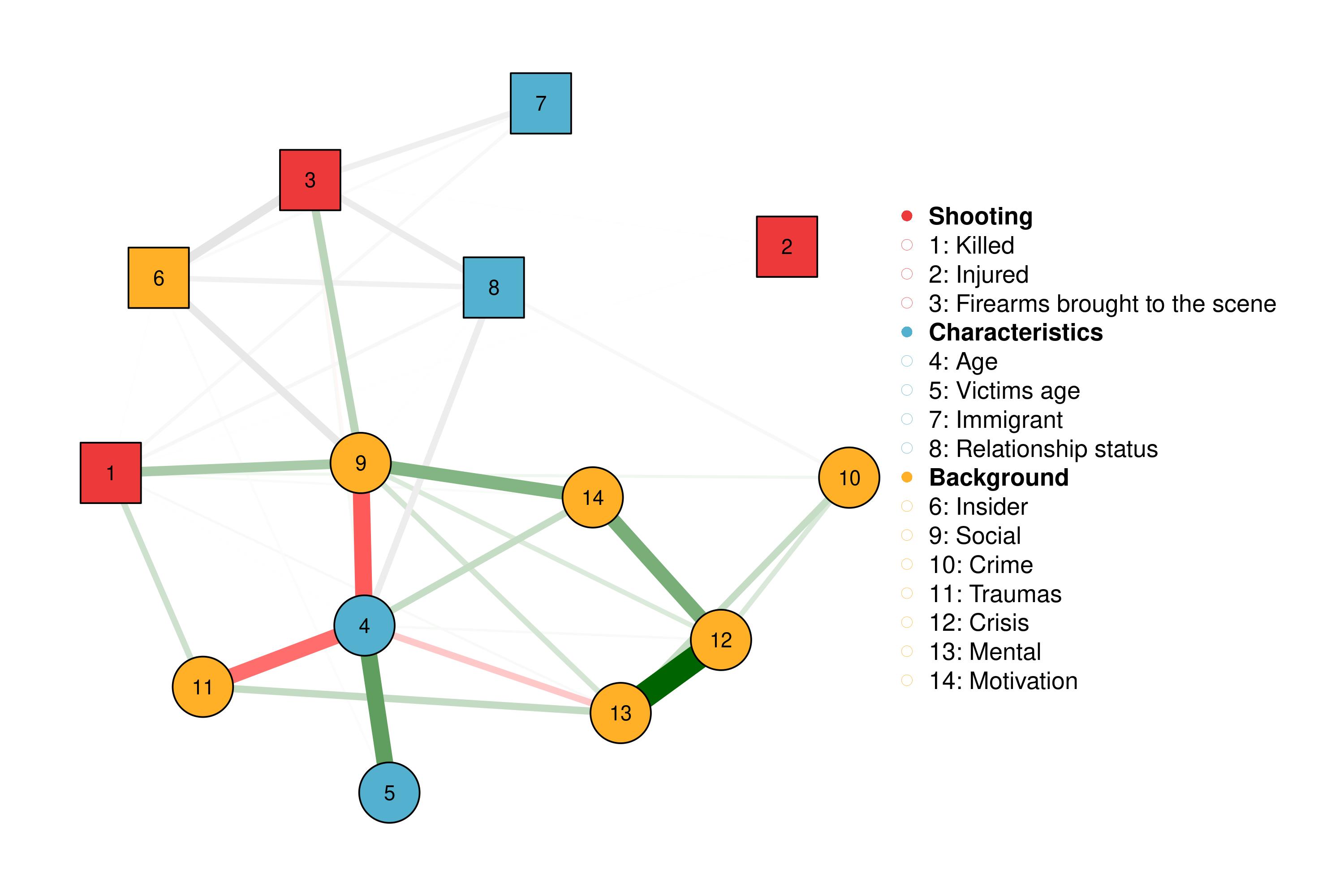}
\caption{Estimated network structures from QMGM7 (top) and MGM (bottom). Green edges in the networks depict positive interactions, red edges represent negative interactions, and thicker/more saturated edges depict stronger interactions.}
\label{fig:graph}
\end{center}
\end{figure}

\begin{figure}[h!]
\begin{center}
\centering
\includegraphics[scale=.425]{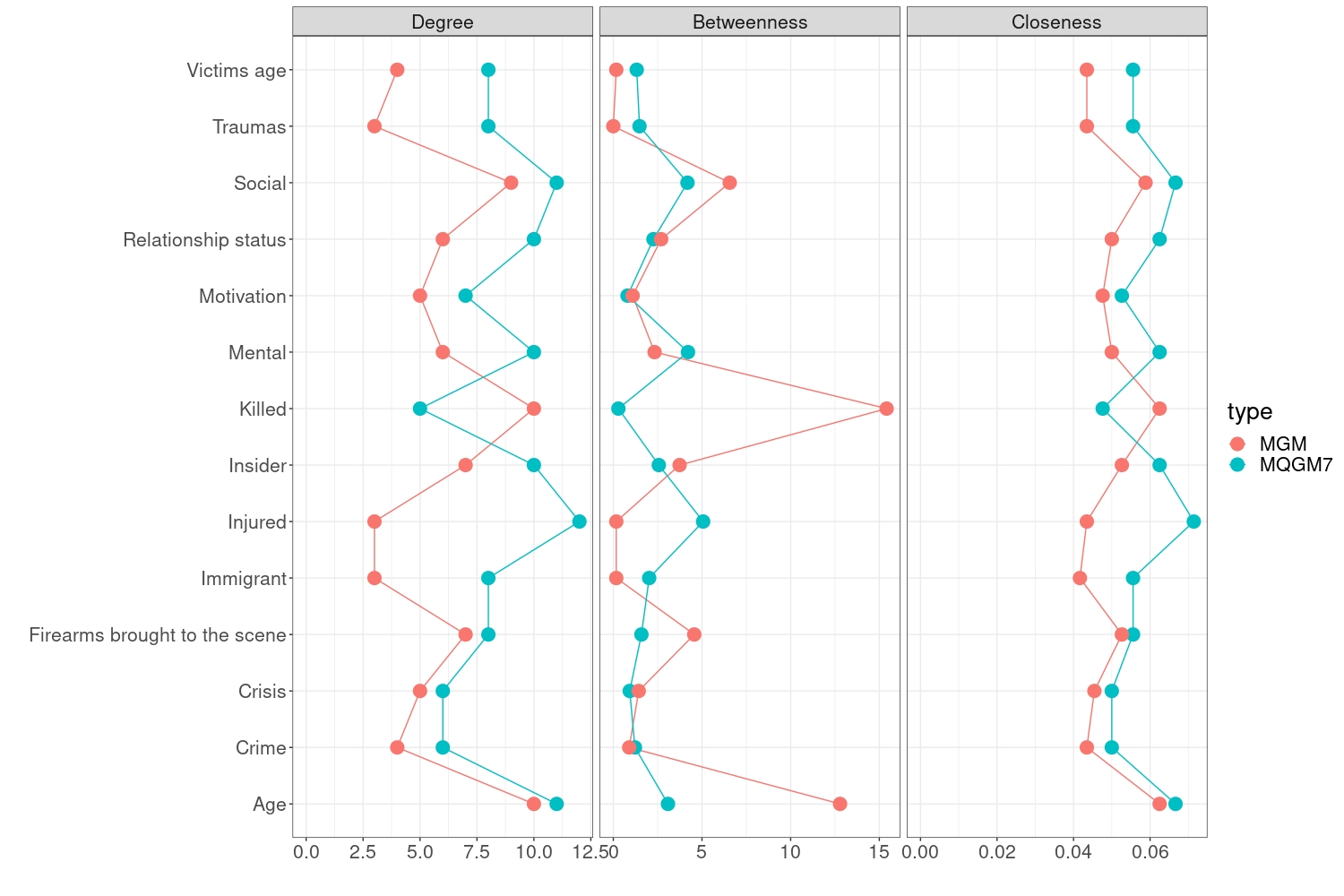}
\caption{Local centrality measures for the MGM (red) and QMGM7 (blue) estimated networks.}
\label{fig:measures}
\end{center}
\end{figure}


\section{Conclusions}\label{sec:con}
Mass shootings are a complex and multifaceted problem that tragically affect {\color{black}people's lives in many countries and, especially,} the US, with incidents happening in schools, churches, movie theaters, workplaces and {\color{black}other public} places. These violent events not only claim lives but also leave lasting emotional and psychological scars on the survivors, the families of the victims and the entire society. Addressing the factors that contribute to these {\color{black}heinous acts} {\color{black}might} help {\color{black}prevent or lessen} the lethality of future mass shootings. {\color{black} Our findings lead to the notion that gun crimes, mental illnesses, social networks, and past traumas are all complex factors that contribute to mass shootings in the US. The {\color{black}graph estimated using our methods suggests} that the connections between the considered entities and gun violence are {\color{black}more} complex and multifaceted {\color{black}than what a simpler mean-based graph model leads us to believe.} This {\color{black}important result} demand{\color{black}s for} a multi-sectoral approach for analysis and {\color{black}investigation as the psychological, sociological and criminal law perspectives would suggest \citep{peterson2021violence}}. Our results are in line with previous studies, supporting the crucial role of individual circumstances such as the characteristics of the person and their life experiences. From a policy perspective, the proposed graph-based framework is also designed as a visualization tool to support decision making. In {\color{black}the fight against} mass shootings, it becomes necessary the involvement of governments, communities, law enforcement agencies, health professionals, and {\color{black}civilians}. Possible strategies include stricter gun control regulations with background checks that can help prevent individuals with malicious intent or mental health issues from obtaining {\color{black}firearms with ease}. Institutions and authorities can seek to facilitate mental health support and mental health awareness campaigns to potentially prevent individuals from resorting to violence. Lastly, the development of adequate reporting systems to improve the quality of mass shootings data for research and data analysis {\color{black}is fundamental}.
}

In this paper, {\color{black}building on \citeauthor{geraci2022mid}'s (\citeyear{geraci2022mid}) mid-quantile regression models,} we develop{\color{black}ed} a mixed graphical model for identifying conditional independence relations between continuous and discrete variables in a quantile framework using Parzen's definition of mid-quantile \citep{parzen1993change}. {\color{black} The proposed network allows us to infer dependence relations that cannot be fully captured by using conditional mean only, by offering a complete characterization of the conditional distributions of the network variables. To recover the graph structure and identify only the most relevant risk factors,}
 we consider a neighborhood selection approach \citep{meinshausen2006high} in which conditional mid-quantiles of each variable in the network are modeled as a sparse function of all others. Graph structure estimation is divided in two steps. We first apply logistic regression to estimate semi-parametrically conditional mid-probabilities of each node. Then, in the second-step LASSO penalized linear mid-quantile regressions are fitted separately on each node of the network over a finite grid of ordered quantile levels. The proposed {\color{black}quantile} methods offer a robust graphical model for {\color{black} a wide range of scenarios such as multimodality, skewness and heavy tailedness,} as well as a{\color{black}n} easily implementable estimation procedure by exploiting traditional tools for regularized regression analysis.

{\color{black}Our methods} can be extended in several directions. First, while the use of a single penalty parameter controlling the overall amount of shrinkage in the network {\color{black}was dictated by parsimony}, one {\color{black}could apply} a separate penalty {\color{black}to edges of different types}. Second, {\color{black}although} the data {\color{black}revealed} a steep rise of mass shootings {\color{black}over time, we neglected the temporal dimension at our disposal}. {\color{black}A generalization of the proposed QMGM to the dynamic framework may offer useful clues to the} understanding {\color{black}of} how the relationships within the network evolve over time {\color{black}and, thus,} {\color{black}provide the basis for more effective interventions}. {\color{black}Finally, the analysis of MPSs data might benefit from the inclusion of demographic, social, ethnic and economic indicators of the location or geographical area where the shootings occurred \citep{ghio}.}

\bibliographystyle{agsm}
\bibliography{Manuscript}

\end{document}